\documentclass[10pt,journal,letterpaper,twocolumn,twoside,nofonttune]{IEEEtranTCOM}
\usepackage{etex}
\normalsize
\usepackage[T1]{fontenc}
\usepackage{amsmath,amssymb,amsfonts}
\usepackage{mathrsfs}
\usepackage{mathabx}
\usepackage{amsbsy}
\usepackage{graphicx}
\usepackage{graphics}
\usepackage{epstopdf}
\usepackage{algorithm}
\usepackage{algorithmic}
\usepackage{subfigure}
\usepackage{cite}
\usepackage{multirow}
\usepackage{ctable}
\usepackage{tikz,pgfplots}
\usepackage{caption}

\title{ 
Achieving the Uniform Rate Region of General Multiple Access Channels by Polar Coding\\[0.50ex]}

\renewcommand{\markboth}[2]
{\renewcommand{\leftmark}{#1}\renewcommand{\rightmark}{#2}}

\author{
Hessam Mahdavifar, ~\IEEEmembership{Member,~IEEE,}
        Mostafa El-Khamy, ~\IEEEmembership{Senior Member,~IEEE,} \\
        Jungwon Lee, ~\IEEEmembership{Senior Member,~IEEE,}
       and Inyup Kang, ~\IEEEmembership{Member,~IEEE}

}

\newtheorem{theorem}{{Theorem}}
\newtheorem{lemma}[theorem]{{Lemma}}

\newtheorem{corollary}[theorem]{{Corollary}}

\newcommand{\cA}{{\cal A}} 
\newcommand{\cB}{{\cal B}}
 
\newcommand{\cD}{{\cal D}}

\newcommand{\cG}{{\cal G}} 

\newcommand{\cI}{{\cal I}} 
\newcommand{\cJ}{{\cal J}}

\newcommand{\cM}{{\cal M}} 
\newcommand{\cN}{{\nu}}

\newcommand{\cS}{{\cal S}} 
\newcommand{\cT}{{\cal T}}
\newcommand{\cU}{{\cal U}}

\DeclareMathAlphabet{\mathbfsl}{OT1}{ppl}{b}{it} 

\newcommand{\bQ}{\mathbfsl{Q}} 
\newcommand{\bR}{\mathbfsl{R}}

\def\QEDclosed{\mbox{\rule[0pt]{1.3ex}{1.3ex}}} 

\def\QED{\QEDclosed} 

\newcommand{\be}[1]{\begin{equation}\label{#1}}
\newcommand{\ee}{\end{equation}} 
\newcommand{\eq}[1]{(\ref{#1})}

\renewcommand{\leq}{\leqslant}
\renewcommand{\ge}{\geqslant} 
\renewcommand{\geq}{\geqslant}

\newcommand{\script}[1]{{\mathscr #1}}

\renewcommand{\Bbb}{\mathbb}
 
\newcommand{\N}{{\Bbb N}}
\newcommand{\R}{{\Bbb R}}

\newcommand{\Tref}[1]{Theo\-rem\,\ref{#1}}

\newcommand{\Lref}[1]{Lem\-ma\,\ref{#1}}
\newcommand{\Cref}[1]{Co\-ro\-lla\-ry\,\ref{#1}}

\newcommand{\deff}{\mbox{$\stackrel{\rm def}{=}$}}

\newcommand{\Gn}{G^{\otimes n}}

\newcommand{\sX}{\script{X}}
\newcommand{\sY}{\script{Y}}

\newcommand{\shalf}{\mbox{\raisebox{.8mm}{\footnotesize $\scriptstyle 1$}
\footnotesize$\!\!\! / \!\!\!$ \raisebox{-.8mm}{\footnotesize
$\scriptstyle 2$}}}

\newcommand{\AQ}{\cA_{\bQ}^{(\pi)}}
\newcommand{\rp}{\bR^{(\pi)}}

\begin{document}

\maketitle

\begin{abstract}

We consider the problem of polar coding for transmission over $m$-user multiple access channels. In the proposed scheme, all users encode their messages using a polar encoder, while a multi-user successive cancellation decoder is deployed at the receiver. The encoding is done separately across the users and is independent of the target achievable rate. For the code construction, the positions of information bits and frozen bits for each of the users are decided jointly. This is done by treating the  polar transformations across all the $m$ users as a single polar transformation with a certain \emph{polarization base}. We characterize the resolution of achievable rates on the dominant face of the uniform rate region in terms of the number of users $m$ and the length of the polarization base $L$. In particular, we prove that for any target rate on the dominant face, there exists an achievable rate, also on the dominant face, within the distance at most $\frac{(m-1)\sqrt{m}}{L}$ from the target rate. We then prove that the proposed MAC polar coding scheme achieves the whole uniform rate region with fine enough resolution by changing the decoding order in the multi-user successive cancellation decoder, as $L$ and the code block length $N$ grow large. The encoding and decoding complexities are $O(N \log N)$ and the asymptotic block error probability of $O(2^{-N^{0.5 - \epsilon}})$ is guaranteed. Examples of achievable rates for the $3$-user multiple access channel are provided. 

\end{abstract}

\begin{keywords} 
polar code, 
multiple access channel,
uniform rate region, successive cancellation decoding
\end{keywords}

\section{Introduction} 
\label{sec:Introduction}

\noindent 
\PARstart{C}{hannel} polarization was introduced by Ar{\i}kan in the seminal work of \cite{Arikan}. Polar codes are the first family of codes for the class of binary-input symmetric discrete memoryless channels that are provable to be capacity-achieving with low encoding and decoding complexity~\cite{Arikan}. Polar codes and polarization phenomenon have been successfully applied to various problems such as wiretap channels \cite{MV}, data compression \cite{Arikan2,Ab}, broadcast channels \cite{MH}, and bit-interleaved coded modulation (BICM) channels~\cite{BICMpolar}. 


The capacity region of multiple access channels is fully characterized by Ahlswede \cite{A} and Liao \cite{L} for the case that the sources transmit independent messages. However, in this paper, we are only interested in the \emph{uniform rate region}. For a multiple access channel, the uniform rate region is the achievable region corresponding to the case that the input distributions are uniform. The single user counterpart of the uniform rate region is indeed the symmetric capacity. It is well-known that the uniform rate region of a multiple access channel (MAC) can be achieved by single user capacity-achieving codes with time-sharing between the users or with the rate-splitting approach \cite{RU, GRU}.

There has been a series of research on channel polarization and achieving the uniform rate region for MAC \cite{STY,AT2,S,CISS} and for the dual problem of source coding \cite{Arikan3} without the time-sharing or the rate-splitting method. In these works, the encoder core for each user is the Ar{\i}kan's polarization matrix. The encoder does not change in time, as in time-sharing, nor depends on the target rate as in the rate-splitting case. The two-user MAC polarization was first studied in \cite{STY}. It was shown that at least one point on the \emph{dominant face} of the uniform rate region can be achieved by the polar code constructed based on the MAC polarization. Ar{\i}kan proposed a scheme for the two-user source coding problem based on \emph{monotone} chain rule expansions and showed that the uniform rate region is achievable with polar coding \cite{Arikan3}. Although not explicitly shown, it was also hinted in \cite{Arikan3} that this method has duals that could be applied to achieve the capacity of the MAC without time-sharing. This is investigated in \cite{CISS}, where we showed how polar coding with multi-user successive cancellation decoding can be used to achieve the uniform rate region of the two-user MAC, along with methods to improve the finite length performance, and performance comparisons with the time-sharing method. A method for improving the performance of the two-user MAC polar coding with list decoding has been described in \cite{S}. The two-user MAC has been also addressed in the related context of interference networks \cite{WS}, where it has been pointed in \cite{WS} that the results can be generalized to the $m$-user MAC. A method for channel polarization in the general case of $m$-user MAC was studied in \cite{AT2}, and the set of the extremal MACs under this channel polarization were characterized. It was also shown that at least one point on the dominant face can be achieved by this MAC polarization.

The contributions of this paper can be summarized as follows. We prove that the entire uniform rate region for the general $m$-user MAC can be achieved by polar coding, without any time-sharing or rate-splitting between the users. This generalizes the result of \cite{Arikan3} from two-user MACs (shown for the two-source Slepian-Wolf problem) to $m$-user MACs. One main contribution of this paper is providing a concrete alternative proof for the \emph{polar splitting} problem of \cite{WS}. We further characterize the resolution of achievable rates on the dominant face and prove that for any target rate on the dominant face, there exists an achievable rate, also on the dominant face, with our proposed MAC polar code within distance $\frac{(m-1)\sqrt{m}}{L}$ of the target rate, where $L$ is the length of the polarization base. It is shown how the channel polarization can be applied on top of a polarization base in order to construct MAC polar coding schemes with a low-complexity multi-user successive cancellation decoder. As a result, all the rates on the dominant face of the uniform rate region can be achieved by MAC polar coding schemes as the code block length and the length of polarization base grows large.

The rest of this paper is organized as follows. In Section\,\ref{sec:two}, we provide notation conventions and review some background on polar codes and multiple access channels. In Section\,\ref{sec:three}, a framework for multi-user polar transformation is discussed and the notion of MAC polarization is introduced followed by a numerical example. In Section\,\ref{sec:four}, the encoding and decoding of MAC polar codes is discussed and a numerical example is provided along with comparisons with the time-sharing method. In Section\,\ref{sec:five}, we characterize the resolution of achievable rates on the dominant face of the uniform rate region and establish our main result in this regard. The achievability property of the MAC polar coding scheme is shown in Section\,\ref{sec:six}. In Section\,\ref{sec:seven}, detailed discussions about other approaches in \cite{STY,AT2, WS, RU, GRU} to achieve the capacity of the MAC are provided.  Concluding remarks are provided Section\,\ref{sec:eight}.

\section{Preliminaries}
\label{sec:two}

\subsection{Notation Convention \label{Sec:Not}}

In this subsection, some of the common notations used through this paper are defined. 

A binary-input discrete memoryless channel (B-DMC) $W$ with input alphabet $\sX = \left\{0,1\right\}$ and output alphabet $\sY$ is also represented as $W: \sX \rightarrow \sY$. The channel $W$ is specified with the transition probabilities which are also denoted by $W$ with some slight abuse of notation. For any $x \in \sX$ and $y \in \sY$, $W(y|x)$ denote the probability of receiving $y$ assuming $x$ is transmitted. 

For any positive integer $n$, let $[\![n]\!]$ denote the set of positive integers less than or equal to $n$. The parameter $m$ is reserved for the number of users in the multiple access channel model. An $m$-user binary input multiple access channel $W : \sX^m \rightarrow \sY$ is also specified with the transition probabilities denoted by $W$, where $\sX^m$ is the Cartesian product of $\sX$ with itself $m$ times. The elements of $\sX^m$ are represented as $m$-tuples $(x[1],x[2],\dots,x[m])$, where $x[j] \in \sX$ denote the input to the channel by $j$-th user, for $j \in [\![m]\!]$. Then for any $y \in \sY$, $W(y|x[1],x[2],\dots,x[m])$ denote the probability of receiving $y$ given that $x[j]$ is transmitted by the $j$-th user, for $j \in [\![m]\!]$. For any $\cJ \subseteq [\![m]\!]$, let $x[\cJ]$ denote the set $\left\{x[j]: j \in \cJ\right\}$. Also, $\cJ^c$ denote the complement of $\cJ$ in $[\![m]\!]$. 

Following the convention, random variables are denoted by upper case letters and their instances are denoted by lower case letters, except for $n$ and $N$ which are reserved to specify the code's block length and $l$ and $L$ which are reserved to specify the length of polarization base. For instance, the input to the channel by the $j$-th user can be a uniform binary random variable $X[j]$ and an an instance of this random variable is denoted by $x[j]$. 

An ordered sequence $\left\{x_i\right\}_{i=1}^N$ is represented by $x_1^N$ for ease of notation. 
The sequence $x_1^N$ is also regarded as a vector of length $N$ for the purpose of matrix operations at the encoder which shall be clear from the context.
Similarly, for any $1\leq i \leq j \leq N$, $x_i^j$ denote the sequence $x_i,x_{i+1},\dots,x_j$. This notation is used for input and output sequences of the channel. Since we will be dealing with $m$-user MACs, the sequence $x_1^N[j] = \left\{x_i[j]\right\}_{i=1}^N$ denotes the input sequence to the channel by the $j$-th user, and the $m$-tuple $(x_i[1],x_i[2],\dots,x_i[m])$ denotes the input to the channel by the $m$ users at the $i$-th channel use. Let  $\left(x_1^N[1]\,x_1^N[2]\,\dots\,x_1^N[m]\right)$ simply denote the sequence of length $mN$, formed by concatenating $m$ sequences each of length $N$. 
For notational convenience, we will let $\cM = [\![m]\!]$, and refer to the sequence formed by the concatenation $x_1^N[j]$ for $j \in \cM$ as $x_1^N[\cM]$.

For a positive integer $N$, a permutation $\pi$ is a bijection from $[\![N]\!]$ to $[\![N]\!]$  and for any $i \in [\![N]\!]$, $\pi(i)$ is the image of $i$. Consecutively, $\pi^{-1}$ is the inverse function of $\pi$, which is another permutation, such that for any $i \in [\![N]\!]$, $\pi(\pi^{-1}(i)) = i$. The set $[\![N]\!]$ over which the permutation $\pi$ is defined shall be clear from the context. The permutation $\pi$ can operate on a sequence $x_1^N$ to relocate $x_i$ to the position indexed by $\pi(i)$. The operation of $\pi$ on the sequence $x_1^N$ is denoted by $\pi \circ x_1^N$. Then by definition we have
$$
\pi \circ x_1^N = \left\{x_{\pi^{-1}(i)}\right\}_{i=1}^N.
$$
For a permutation $\pi$ on $[\![\!N]\!]$ and any $k$, where $k,N \in \N$, we define the permutation $\pi^{(k)}$ on $[\![kN]\!]$, as follows. For any $i \in [\![kN]\!]$, let $i = k(r-1)+q$, where $q \in [\![k]\!]$ and $r \in [\![N]\!]$. Then $\pi^{(k)}(i) = k(\pi(r)-1)+q$. For instance, if $N = 2$ and $\pi(1) = 2, \pi(2) =1$, then for any $k$, the permutation $\pi^{(k)}$ replaces the first and second sub-sequence of length $k$ with each other.

The points in $\R^m$ are represented by bold upper case letters. A point $\bR \in \R^m$ can be also represented by its $m$ coordinates $R_1,R_2,\dots,R_m$ as $(R_1,R_2,\dots,R_m)$. 

\subsection{Polar Codes}

In this subsection, we provide a brief overview of polar codes and channel polarization ~\cite{Arikan, AT,Korada,KSU}. 

The \emph{channel polarization} phenomenon was discovered by Ar{\i}kan \cite{Arikan}. 
A basic binary polarization matrix is given by
\be{G-def}
G 
\ = \
\left[ 
\begin{array}{cc}
1 & 0\\
1 & 1\\ 
\end{array}
\right].
\ee

Consider two independent copies of a B-DMC $W$. The two input bits $(u_1,u_2)$, drawn from independent uniform distributions, are multiplied by $G$ and then transmitted over the two copies of $W$. One level of channel polarization is the mapping $(W,W) \rightarrow (W^{-}, W^{+})$, where 
$W^-: \{0,1\} \to \sY^2$, and  $W^+: \{0,1\} \to \{0,1\} \times \sY^2$ with the following channel transformation 
\begin{align}
\label{channel-com1}
&W \boxcoasterisk W(y_1,y_2|u_1) = \frac{1}{2}  \sum_{u_2 \in \{0,1\}} W(y_1 | u_1 \oplus u_2) W(y_2|u_2), \\
\label{channel-com2}
&W \circledast W(y_1,y_2,u_1|u_2) =  \frac{1}{2} W(y_1 | u_1 \oplus u_2) W(y_2|u_2).
\end{align}
$W \boxcoasterisk W$ and $W \circledast W$ are also denoted by $W^+$ and $W^-$. 

The channel polarization is continued recursively by further splitting $W^{-}$ and $W^{+}$ . This process can be explained best by means of Kronecker powers of $G$. Let $G^{\otimes 1} = G$ and for any $n > 1$:
$$
G^{\otimes n}
\ = \
\left[ 
\begin{array}{cc}
G^{\otimes (n-1)} & 0\\
G^{\otimes (n-1)} & G^{\otimes (n-1)}\\ 
\end{array}
\right]
$$
Let $N = 2^n$. Then $\Gn$ is an $N \times N$ polarization matrix. Let $U_1^N$ be a sequence of $N$ independent and uniform binary random variables. The polarization matrix $\Gn$ is multiplied by $U_1^N$ to get $X_1^N$. Then $X_i$'s are transmitted through $N$ independent copies of a B-DMC $W$. The output is denoted by $Y_1^N$. This transformation with input $U_1^N$ and output $Y_1^N$ is called the polar transformation. 

In the polar transformation, $N$ independent uses of $W$ are transformed into $N$ bit-channels, which are the channels that the encoded bits observe through the successive cancellation decoding.  Ar{\i}kan used the \emph{Bhattacharyya parameter} of $W$, denoted by $Z(W)$, to measure how good the binary-input channel $W$ is
$$ 
Z(W)
\,\ \deff\kern1pt
\sum_{y\in\sY} \!\sqrt{W(y|0)W(y|1)}.
$$

The channel polarization theorem is proved by showing that the fraction of good bit-channels, i.e., the bit-channels with Bhattacharyya parameter less than a certain threshold, approaches the symmetric capacity of $W$ as $N$ goes to infinity \cite{AT}. Then, for polar code construction, the idea is to transmit the information bits over the good bit-channels while freezing the input to the other bit-channels to a priori known values, such as zeros.

\subsection{Uniform Rate Region}

Let $X[1],X[2],\dots,X[m]$ be independent and uniform binary random variables. Then, the uniform rate region of $W$, denoted by $\cU(W)$, is defined to be the set of all points $\bR = (R_1,R_2,\dots,R_m) \in \R^m$ such that
\be{region-def}
0 \leq \sum_{j \in \cJ} R_j \leq I(X[\cJ];Y,X[\cJ^c]), \; \forall \cJ \subseteq [\![m]\!], 
\ee
where $I(\cdot;\cdot)$ represents the mutual information.

The uniform rate region is the set of all achievable points $\bR$ assuming that the input distributions are uniform.  The \emph{uniform sum-rate} of $W$, $\cI(W)$, is defined as follows:
$$
\cI(W) = I(X[1],X[2],\dots,X[m];Y)
$$ 
In general, in the context of this paper, any point $\bR = (R_1,R_2,\dots,R_m) \in \R^m$ is regarded as an $m$-tuple of rates where $R_j$ denotes the rate of the $j$-th user. Also, $\sum R_j$,  is referred to as the sum-rate of $\bR$. 

The dominant face of the uniform rate region, denoted by $\cD(W)$, is defined to be the set of points in $\cU(W)$, with the maximum sum-rate $\cI(W)$, i.e., the right inequality of \eq{region-def} is in fact equality for $\cJ = [\![m]\!]$. It can be observed that, the achievability of the uniform rate region $\cU(W)$ is equivalent to the achievability of its dominant face $\cD(W)$. Therefore, our focus throughout this paper is on the achievability of $\cD(W)$. 

\section{Multi-user Polar Transformation: Framework and Example}
\label{sec:three}

In this section, we discuss the framework considered in this paper for extending polar transformation concept from the single user case to the multi-user case under the MAC model. We also introduce the concept of MAC polarization base and provide a numerical example for a $3$-user MAC.

\subsection{MAC Polar Transformation}

Let $W$ be a given $m$-user binary-input discrete multiple access channel. Let also $n$ be a positive integer and $N = 2^n$. For $j = 1,2,\dots,m$, assume that $U_1^N[j]$ is a sequence of $N$ independent and uniformly distributed bits that represents the $j$-th user's message, and is independent of other users' messages. Let also $X_1^N[j] = U_1^N[j]\,G^{\otimes n}$. For $i=1,2,\dots,N$, the $m$-tuple $(X_i[1],X_i[2],\dots,X_i[m])$ is transmitted through the $i$-th independent copy of $W$ and the output is denoted by $Y_i$. 
Let $\pi$ be a permutation on $[\![mN]\!]$ to be applied on the concatenated sequence of input messages to give
$$
D_1^{mN} = \pi \circ \left(U_1^N[1]\,U_1^N[2]\,\dots\,U_1^N[m]\right)
$$
The permuted sequence $D_1^{mN}$ specifies the order in which the input bits are decoded by the multi-user successive cancellation decoder, as will be specified later in this section. 
The transformation between the input sequences $U_1^N[j]$ for $j \in [\![m]\!]$ and the output sequence $Y_1^N$ together with the permutation $\pi$, which enforces the decoding order, is called the MAC polar transformation. Also, $N$ is referred to as the length of the transformation. 


Let $W^N: \sX^{mN} \rightarrow \sY^N$ denote the channel consisting of $N$ independent copies of $W$, i.e., 
\be{Wnm}
W^N\kern-0.5pt(y^N_1|x_1^N[\cM]) 
\,\ \deff\,\
\prod_{i=1}^N W(y_i|x_i[\cM]).
\vspace{-0.25ex}
\ee
 For a given $N$, the combined channel $\widetilde{W}$ is defined with transition probabilities given by
\be{Wtildem}
\widetilde{W}(y^N_1|u_1^N[\cM]) 
\,\ \deff\,\ W^N\kern-0.5pt(y^N_1|x_1^N[\cM]),
\ee
where $x_1^N[j] = u_1^N[j]\Gn$, for $j \in \cM$. The bit-channels are defined with respect to the ordered sequence $d_1^{mN}$. For $i=1,2,\dots,mN$, the $i$-th bit-channel is defined as
\be{Wi-def-general}
W^{(i)}_{N}\bigl( y^N_1,d^{i-1}_1 | \hspace{1pt}d_i)\,\ \deff \,\
\frac{1}{2^{mN-1}}\hspace{-5pt}
\sum_{d_{i+1}^{mN} \in \{0,1\}^{mN-i}} \hspace{-12pt}
\widetilde{W} \Bigl(y^N_1\hspace{1pt}{\bigm|}\hspace{1pt}
d_1^{mN} \Bigr).
\ee

Let $D_1^{mN}$ be a priori uniform over $\{0,1\}^{mN}$. Then, it can be shown that  $W^{(i)}_{N}\bigl( y^N_1,d^{i-1}_1 | \hspace{1pt}d_i)$ is indeed the probability of the event $Y_1^N = y_1^N$ and $D_1^{i-1} = d_1^{i-1}$, given~the event $D_i = d_i$. 

Next, we introduce the notion of MAC polarization base. Let $L = 2^l$, where $l \leq n$ is fixed. Consider a permutation $\pi$ on $[\![mL]\!]$ and a MAC polar transformation of length $L$ associated with permutation $\pi$. This is referred to as the MAC polarization base associated with $\pi$ and $L$ is referred to as the length of the polarization base. The reason to introduce a new notion is that we fix the length $L$ of the MAC polarization base while letting $N$ to grow large in order to establish the MAC polarization theorem in Section \ref{sec:six}. To this end we consider MAC polar transformation of length $N$ associated with $\pi^{(N/L)}$, as defined in Section \ref{Sec:Not}. The MAC polar transformation associated with $\pi^{(N/L)}$ is said to be built upon the MAC polarization base of length $L$ associated with $\pi$, a relation which will be clarified in through the explanation of the decoder implementation in Section\,\ref{sec:four} and the MAC polarization in Section \ref{sec:six}. 

\subsection{Input-to-Output Mutual Information}

Consider MAC polarization bases of length $L$ with all possible permutations $\pi$ on the set $[\![mL]\!]$. Then $U_1^L[j]$ represents the $j$-th user's sequence of input bits and $Y_1^L$ denote the output sequence.  

For a given $\pi$, let $D_1^{mL}$ denote the permuted sequence of the concatenated input sequence $U_1^L[1],U_1^L[2],\dots,U_1^L[m]$. The mutual information $I(D_1^{mL};Y_1^L)$ can be expanded using the chain rule of mutual information as follows:
$$
I(D_1^{mL};Y_1^L) = \sum_{i=1}^{mL} I(D_i;Y_1^L,D_1^{i-1})
$$

Then for $i=1,2,\dots,mL$, let $I_i^{(\pi)}$ denote the $i$-th term in the above expansion, i.e.,
\be{Iji-dfn}
I_i^{(\pi)}\,\deff\, I(D_i;Y_1^L,D_1^{i-1}).
\ee
Also, for $j = 1,2,\dots,m$, let
\be{Rj-dfn}
R_j^{(\pi)}\ \deff\ \frac{1}{L} \sum_{i=(j-1)L+1}^{jL} I_{\pi(i)}^{(\pi)}.
\ee
In fact, $R_j^{(\pi)}$ is the input-to-output mutual information of $j$-th user normalized by the transformation length $L$. This parameter is also referred to as the allocated rate to the $j$-th user. As we will see in the Section\,\ref{sec:six}, the $m$-tuple of rates $\bR^{(\pi)} = (R_1^{(\pi)},R_2^{(\pi)},\dots,R_m^{(\pi)})$ is achievable with MAC polar coding built upon the polarization base of length $L$ and permutation $\pi$.  

The following lemma is the result of \eq{Iji-dfn} and \eq{Rj-dfn}.
\begin{lemma}
\label{sum-rate}
For any permutation $\pi$, the sum-rate of $\bR^{(\pi)}$ is $\cI(W)$.
\end{lemma}
\noindent
{\bf Remark.} \looseness=-1 
In fact a more general statement than \Lref{sum-rate} holds. For any permutation $\pi$, $\bR^{(\pi)}$ is a point on the dominant face $\cD(W)$. This will also follow as a result of the Section\,\ref{sec:six}, where we prove that these points are achievable by polar coding.

\subsection{MAC Polarization Bases for a $3$-User Case}

In this subsection, we consider the special case of polar coding for $3$-user multiple access channels, $m = 3$. Let $W$ be a $3$-user MAC. The uniform rate region of $W$ is in general a $3$-dimensional polyhedron, which is shown for an example, described later, in Fig.~\ref{region2}, where $R_1$, $R_2$ and $R_3$ are rates of user $1$, $2$ and $3$, respectively. 

Let $X[1]$, $X[2]$ and $X[3]$ be uniform and independent binary random variables which are the inputs to $W$ and let $Y$ denote the output. Let $I_1$, $I_2$ and $I_3$ be defined as follows:
\begin{align*}
I_1 &= I(X[1];Y)\\
I_2 &= I(X[2];Y,X[1])\\
I_3 &= I(X[3];Y,X[1,2]).
\end{align*}
The dominant face $\cD(W)$ is the hexagon whose vertices are specified in Fig.~\ref{face}. These vertices are actually the corner points of the uniform rate region. For simplicity, it is assumed that $W$ is symmetric with respect to the inputs, i.e., $W(y|x[1],x[2],x[3])$ remains the same if $(x[1],x[2],x[3])$ is permuted.

\begin{figure}[h]
\centering
\includegraphics[width=\linewidth]{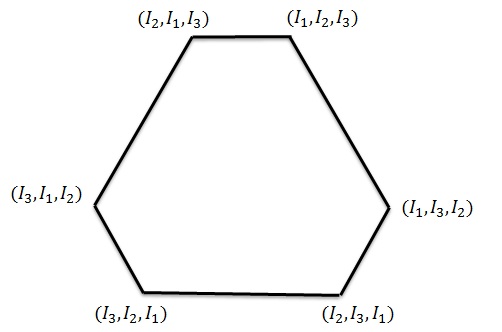}
\caption{The dominant face $\cD(W)$.}
\label{face}
\end{figure} 

The corner points are achievable by polar coding separately across the users. For instance, in order to achieve the point $(I_1,I_2,I_3)$, the message of user $1$ is decoded first. Then the message of user $2$ is decoded assuming the user $1$'s message is known. At the end, the message of user $3$ is decoded, assuming the first and second messages are known. In this case, the scheme does not depend on the underlying polar codes and any other capacity achieving code will fit as well. 


We consider polarization bases of length $2$. An example of a polarization base of length $2$ is shown in Fig.~\ref{building3}. The decoding order of this polarization base is given by $x_1[1],x_1[2],x_1[3], x_2[1], x_2[2], x_2[3]$. This decoding order can be simply denoted by a permutation $\pi$ which permutes the default ordered sequence $x_1[1],x_2[1],x_1[2], x_2[2], x_3[1], x_3[2]$. In this case, $\pi(1) = 1, \pi(2)=4, \pi(3)=2, \pi(4)=5, \pi(5) = 3, \pi(6) = 6$. For polar coding with a general block length $N$ built upon this polarization base, the decoding order is specified as follows. The successive cancellation decoder decodes the first half of user $1$'s message first, then the first half of user $2$'s message and then the first half of user $3$'s message followed by the second half of the messages of user $1$, $2$ and $3$, respectively.

\begin{figure}[h]
\centering
\includegraphics[width=\linewidth]{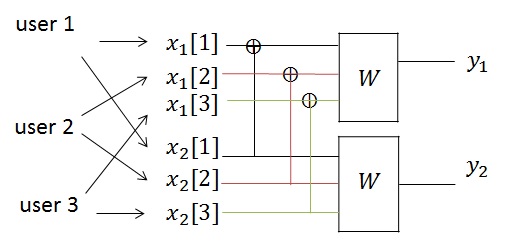}
\caption{An example for a polarization base of length $2$.}
\label{building3}
\end{figure} 

In total, there are $6! = 720$ possible polarization bases of length $2$. It is worth to emphasize that our scheme does not require a \emph{monotone} decoding order, in the sense defined in~\cite{Arikan3}, for each of the users, i.e. $x_2[j]$ may appear before $x_1[j]$ in the permuted sequence. For each polarization base, associated with a permutation $\pi$, one can derive the input-to-output mutual information for each of the users that constitute the $3$-tuple of rates $\bR^{(\pi)} \in \R^3$. This is done for a binary-additive Gaussian noise channel $W$. In this model, the input bits $x[1], x[2], x[3]  \in \left\{0,1\right\}$ are modulated using BPSK ($0$ is mapped to $-1$ and $1$ is mapped to $+1$) into $\overline{x}[1]$, $\overline{x}[2]$ and $\overline{x}[3]$, respectively. The output of the channel is denoted by $y$, where $y=\overline{x}[1]+\overline{x}[2]+\overline{x}[3]+\cN$ and $\cN$ is the Gaussian noise of variance $1$. For this channel, it is well-known that the capacity region is the same as the uniform rate region (see for example \cite[Chapter 5.5.]{modern}). This is given by the set of all possible $3$-tuple rates $(R_1,R_2,R_3)$ that satisfy
\begin{align*}
0 \leq R_1,R_2,R_3 &\leq I_3  = 0.7215\\
R_1+R_2,R_1+R_3,R_2+R_3 &\leq I_3+I_2 = 1.1106\\
R_1+R_2+R_3  &\leq I_1+I_2+I_3 =  1.3681
\end{align*}

We have numerically computed all the $720$ possible points $\bR^{(\pi)}$ for all the possible permutations $\pi$ over $\left\{1,2,\dots,6\right\}$. However, some the permutations results in the same point $\bR^{(\pi)}$. We have identified $474$ distinct points on the dominant face which are shown in Fig.~\ref{region2}. We will prove in Section\,\ref{sec:six} that these points are achievable with MAC polar coding built upon the polarization bases of length $2$, as the code block length grows large. 

\begin{figure}[h]
\centering
\includegraphics[width=\linewidth]{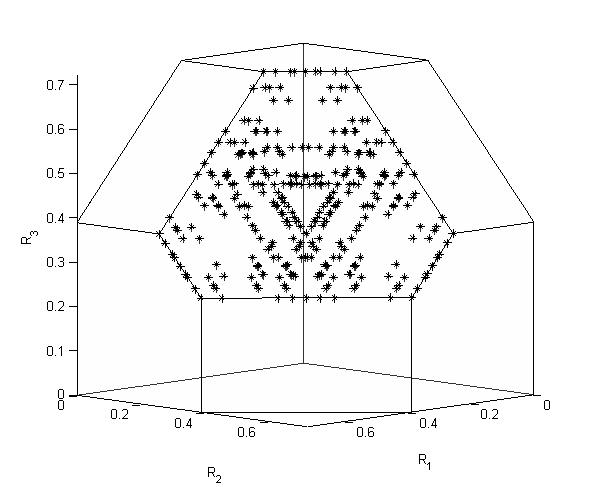}
\caption{The uniform rate region and the achievable points with polarization bases of length $2$.}
\label{region2}
\end{figure} 

\section{MAC Polar Codes: Encoding and Decoding}
\label{sec:four}

In this section, we discuss the construction of MAC polar codes, the encoding process and the multi-user successive cancellation decoding along with a low-complexity implementation using our proposed scheme. The MAC polar code construction is done numerically for a $2$-user case and the achievable rates are compared with that of a conventional time-sharing approach. 

\subsection{Construction and Encoding}
\label{sec:construction}

The construction of MAC polar codes follow a similar logic to that of the single-user polar code construction \cite{Arikan}. For each user, some of the input bits will carry information bits while the rest of input bits will be assigned a priori known frozen values. The choice of indices which will be assigned information bits will depend on the bit-channels of the MAC polar transformation defined in the previous subsection.

For a given threshold $\cT$, where $0 < \cT < 1$, the set of good bit channels, that will carry the information bits, constitutes of the bit-channels with a Bhattacharyya parameter less than $\cT$. Similar to the single-user polar codes, the parameter $\cT$ can be used to derive an upper bound on the error probability of the decoder. Since the error probability of bit-channels is upper bounded by their Bhattacharyya parameters, the total error probability is then upper bounded by $mN\cT$, using the union bound. The set of good bit-channels for the $j$-th user is denoted by $\cG_N^{(j)}(W,\cT)$ and is defined as follows:
\begin{equation}
\label{good-def}
\cG_N^{(j)}(W,\cT)\ {\deff}\ \left\{\, i \in s_{\pi}[j]~:~ Z(W^{(i)}_{N}) < \cT \hspace{1pt}\right\},
\end{equation}
where  $s_{\pi}[j]$ denotes the set of positions of the input bits of the $j$-th user in the ordered sequence specified by $\pi$, i.e.,
\begin{equation}
\label{S-def}
s_{\pi}[j]\,\deff\, \left\{\pi(i) ~:~ N(j-1)+1 \leq i \leq Nj \right\}.
\end{equation}

Then, the MAC polar code construction proceeds as follows. For any $j \in [\![m]\!]$ and $i \in [\![N]\!]$, if $(j-1)N+i \in \cG_N^{(j)}(W,\cT)$, then $u_i[j]$ is assigned an information bit for $j$-th user. Otherwise, $u_i[j]$ is frozen to a fixed value known to both the encoder and the decoder. If the underlying MAC is symmetric for any of the users, then the frozen values can be set to any arbitrary value, e.g., all zeros. Otherwise, the frozen values are chosen independently at random and revealed to the decoder a priori. In that case, the probability of error will be calculated as the average over all the possible information vectors as well as the frozen values. In the numerical examples provided in this paper, the underlying channels are symmetric and hence, the frozen values are set to zeros. 

The encoding for the MAC polar scheme is similar to that of original polar codes for each of the users. In fact, the encoder for any block length $N$ is fixed, where each user multiplies the resulting vector $u_1^N[j]$ by the polarization matrix $\Gn$ regardless of the choice of the permutation $\pi$ and the threshold $\cT$ . Then, the individual code rate of the $j$-th user is given by $\frac{|\cG_N^{(j)}(W,\cT)|}{N}$.

\subsection{Multi-user Successive Cancellation Decoding}
\label{sec:decoding}

In this subsection, the multi-user successive cancellation decoding of the MAC polar codes is discussed. We then discuss how a low complexity implementation of the decoder is possible when the MAC polar transformation is built upon a MAC polarization base. 

The successive cancellation decoding of single-user polar codes \cite{Arikan}, is invoked to decode the MAC polar codes as well. The permuted sequence $d_1^{mN}$ is decoded successively regardless of which user the bit $d_i$ belongs to. The decoder attempts to estimate the $i$-th bit $d_i$ having 
observed $y_1^n$ and estimated $d_1^{i-1}$. If $i \notin \cG_N^{(j)}(W,\cT)$, for $j \in [\![m]\!]$, then $d_i$ is set the corresponding frozen value which has been revealed to the decoder a priori. Otherwise, the optimal decision rule for the decoder is to decide $d_i = 0$ if 
\be{decision-rule}
W^{(i)}_{N}\bigl( y_1^N,d_1^{i-1} | 0) 
\, \ge \,
W^{(i)}_{N}\bigl( y_1^N,d_1^{i-1} | 1), 
\ee
and $d_i = 1$ otherwise. This hard-decision rule is invoked successively for all $i = 1,2,\dots,mN$. 

The decoder is aware of the permutation $\pi$ and the sets of good bit-channels $\cG_N^{(j)}(W,\cT)$. Hence, after successive decoding of $d_1^{mN}$, the multi-user decoder can recover the individual message transmitted by each of the users.  

Next we explain a low complexity implementation of the multi-user successive cancellation decoder assuming our proposed scheme of MAC polar transformation built upon a MAC polarization base of a fixed length $L$. To this end we consider MAC polar transformation of length $N$ associated with $\pi^{(N/L)}$, as defined in Section \ref{Sec:Not}. Consider a log likelihood ratio (LLR) calculator sub-block $\cS_1$ associated with the MAC polarization base corresponding to the permutation $\pi$ that successively calculates the soft information as in \eq{decision-rule} given $L$ channel outputs. This can be done using a naive way with complexity $O(2^{mL})$ for each LLR. Note that this is the worst-case complexity in the sense that depending on the choice of the particular base one can invoke the recursive structure, imposed by $G^{\otimes l}$ for each user, to compute the LLRs more efficiently. Also, consider a sub-block $\cS_2$ that implements a successive cancellation decoder trellis of size $(N/L)\times (n-l+1)$ with $N/L$ inputs and $N/L$ outputs, as originally proposed in~\cite{Arikan}.

The low-complexity multi-user successive cancellation decoder is a concatenation of two major blocks. The first block constitutes of $N/L$ parallel sub-blocks $\cS_1$ and the second block constitutes of $mL$ sub-blocks $\cS_2$. The $N$ channel observations are the input to the first block. The $N/L$ sub-blocks $\cS_1$ are processed in parallel, with $L$ channel observations each, and the first calculated LLRs of $\cS_1$'s are fed into the first $\cS_2$. After the process of the first $\cS_2$ is finished, the output hard decisions are used for each of the $\cS_1$'s to proceed the successive cancellation decoding and so on. This recursive structure follows from the particular choice of the permutation $\pi^{(N/L)}$ that is a lifted version of $\pi$ and it will be shown more explicitly through the recursive MAC channel polarization in Section\,\ref{sec:six}. The total decoding complexity is then upper bounded by $O(mN(n-l+1+2^{mL}))$, where $m$ and $L$ are regarded as fixed parameters in the scheme. If $L$ is fixed and $N$ grows large, the decoding complexity is asymptotically $O(N \log N)$, similar to the original polar decoding.

\subsection{Code Construction for a $2$-User MAC}

In this subsection, a binary-additive two-user Gaussian channel $W$ is picked for the simulation model similar to the model in the previous section. The Gaussian noise has variance $1$. We then consider $6$ different MAC polar transformations with polarization bases of length $2$. The multi-user successive cancellation decoding is implemented for each of MAC polar coding schemes in order to estimate the probability of error of the bit-channels. Then the code is constructed assuming the total frame error rate of the multi-user decoder to be $10^{-2}$. We construct MAC polar codes for the $6$ different cases with block length $N=1024$ and $N=4098$. In order to show the advantage of our MAC polar coding scheme in comparison with a time-sharing method, time sharing between the two corner points, to get their midpoint, with underlying point-to-point polar codes are considered. In this case, the total code block length for each of the users is the same $N = 1024, 4098$ which constitutes of two separate codes between which the time-sharing is done. The code rates are calculated assuming the frame error rate $10^{-2}$. The resulting rates are shown in Fig.~\ref{plot1}. One can observe that our proposed MAC polar coding scheme offers improved rate comparing to the time-sharing method at the same total code block length and frame error rate. Note that this improvement comes at the cost of additional complexity of calculating the likelihood ratios across the polarization base in the decoder which does not exist for the time-sharing method.  

\begin{figure}[h]
\centering
%

\begin{tikzpicture}

\definecolor{mycolor1}{rgb}{0,1,1}

\begin{axis}[
scale only axis,
width = 2.75in,
height = 2in,
xmin=0, xmax=1,
ymin=0, ymax=1,
xlabel={$R_1$},
ylabel={$R_2$},
legend entries={{\scriptsize $N=2^{10}$ MAC-polar},{\scriptsize $N=2^{10}$ time-sharing}, {\scriptsize $N=2^{12}$ MAC-polar},{\scriptsize $N=2^{12}$ time-sharing},{\scriptsize Capacity region}},
axis on top]

\addplot [
color=red,
dashed,
mark=+,
mark options={solid}
]
coordinates{
 (0,0.5352)
 (0.2217,0.5352)
 (0.2617,0.5098)
 (0.3682,0.4033)
 (0.4033,0.3682)
 (0.5098,0.2617)
 (0.5352,0.2217)
 (0.5352,0)

};

\addplot [
color=green,
dashed,
mark=o,
mark options={solid}
]
coordinates{
 (0,0.4961)
 (0.1934,0.4961)
 (0.3448,0.3448)
 (0.4961,0.1934)
 (0.4961,0)

};

\addplot [
color=blue,
dashed,
mark=asterisk,
mark options={solid}
]
coordinates{
 (0,0.5771)
 (0.2585,0.5771)
 (0.2993,0.5469)
 (0.4133,0.4329)
 (0.4329,0.4133)
 (0.5469,0.2993)
 (0.5771,0.2585)
 (0.5771,0)

};

\addplot [
color=mycolor1,
dashed,
mark=x,
mark options={solid}
]
coordinates{
 (0,0.5479)
 (0.2344,0.5479)
 (0.3912, 0.3912)
 (0.5479,0.2344)
 (0.5479,0)

};

\addplot [
color=black,
solid
]
coordinates{
 (0,0.7215)
 (0.3888,0.7215)
 (0.7215,0.3888)
 (0.7215,0)

};

\end{axis}

\end{tikzpicture}
\caption{Rates of the constructed MAC polar codes in comparison with time-sharing}
\label{plot1}
\end{figure}
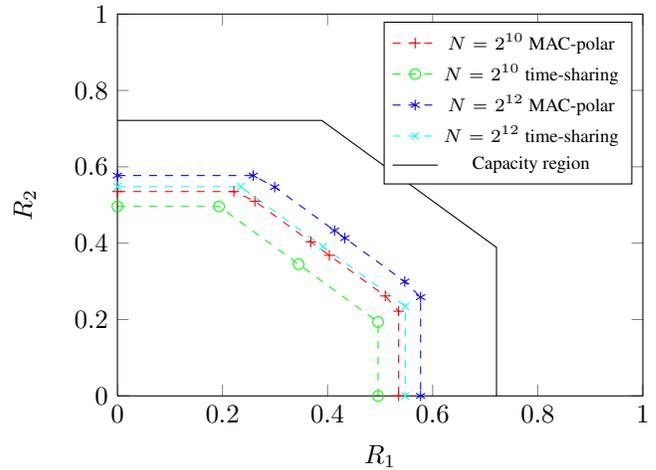 

Furthermore, as discussed in \cite{CISS}, one can invoke the compound polar codes proposed in \cite{compound} in order to improve the finite length performance if time sharing is used.


\section{Achievable Points on the Dominant Face}
\label{sec:five}

In this section, we consider MAC polarization bases of a fixed length $L$ and characterize the input-to-output mutual information for each of the users, also referred to as the allocated rates. Then, it is shown that $m$-tuples of the allocated rates cover the dominant face of the uniform rate region with a resolution characterized by the number of users $m$ and the length $L$. In particular, we prove that as $L$ grows large the entire dominant face is covered with fine enough resolution. 

Let $W$ be a given $m$-user binary-input discrete multiple access channel. Let $L = 2^l$, where $l$ is a positive integer, and consider MAC polarization bases of length $L$ with all possible permutations $\pi$ on the set $[\![mL]\!]$. The transformation length $L$ is fixed through the rest of this section and hence, a permutation is always assumed to be on $[\![mL]\!]$. 

Let $R_j^{(\pi)}$, as defined in \eq{Rj-dfn}, denote the input-to-output mutual information of $j$-th user normalized by the transformation length $L$ and let $\bR^{(\pi)} = (R_1^{(\pi)},R_2^{(\pi)},\dots,R_m^{(\pi)})$. The covering radius $r$, with respect to the transformation length $L$, is formally defined as follows:
\be{covering-radius}
r \,\deff\, \max_{\bQ \in \cD(W)} \min_{\pi} \left\| \bQ - \bR^{(\pi)}\right\|
\ee
where $\left\|.\right\|$ is the Euclidean norm in the $m$-dimensional space $\R^m$. One can think of the covering radius as the smallest $r$ such that the closed spheres of radius $r$ centered at all $m$-tuple rates $\bR^{(\pi)}$ cover the entire dominant face. In other words, for any point $\bQ$ on the dominant face, there exists a permutation $\pi$ such that the distance between $\bR^{(\pi)}$ and $\bQ$ is less than $r$. We think of the covering radius as a measure to characterize the resolution of the points $\bR^{(\pi)}$ on the dominant face.

For two users $j_1$ and $j_2$, we write $j_1 \rightarrow j_2$ if an input bit of the user $j_1$ appears right before an input bit of the user $j_2$ through the permutation $\pi$. More precisely, there exist $i_1$ and $i_2$ with $(j_1-1)L+1 \leq i_1 \leq j_1 L$ and $(j_2-1)L+1 \leq i_2 \leq j_2 L$ such that $\pi(i_2) = \pi(i_1)+1$. In that case, we also define the new permutation $\pi'$ from $\pi$ by swapping $\pi(i_1)$ and $\pi(i_2)$, i.e., $\pi'(i_2) = \pi(i_1)$, $\pi'(i_1) = \pi(i_2)$ and $\pi'(i) = \pi(i)$ for $i \neq i_1,i_2$. We say that $\pi'$ is the transposition of $\pi$ with respect to $j_1 \rightarrow j_2$. Notice that the transposition with respect to $j_1 \rightarrow j_2$ is not necessarily unique, as there may be other input bits of the user $j_1$ that appear right before other input bits of the user $j_2$. In that case, we may choose one of them for the transposition.

The next two lemmas are useful to establish the proof of \Tref{cover} which is the main theorem of this section.
\begin{lemma}
\label{resolution1}
Let $\pi'$ be the transposition of $\pi$ with respect to $j_1 \rightarrow j_2$. Then,
$$
R_j^{(\pi)} = R_j^{(\pi')}\ \text{for}\ j \neq j_1,j_2 
$$
and
$$
0 \leq R_{j_1}^{(\pi')} - R_{j_1}^{(\pi)} = R_{j_2}^{(\pi)} - R_{j_2}^{(\pi')} \leq \frac{1}{L}.
$$
\end{lemma}

The proof can be found in Appendix A.

Fix an arbitrary point $\bQ = (Q_1,Q_2,\dots,Q_m)$ on the dominant face $\cD(W)$. For a given permutation $\pi$, let $\AQ \subset [\![m]\!]$ be defined as follows:
\be{AQ-def}
\AQ \,\deff\, \left\{j\in [\![m]\!]: R^{(\pi)}_j < Q_j\right\}
\ee
\begin{lemma}
\label{resolution2}
For any $\cB \subset \AQ$, at least one input bit from the complement set $\cB^c$ appears after some input bits of the set $\cB$ through the permutation $\pi$. 
\end{lemma}
\begin{proof}
Assume, to the contrary, that all input bits of the users in the set $\cB^c$ appear before all the input bits of the users in the set $\cB$. This, together with the chain rule of mutual information, and the definition of $R_j^{(\pi)}$ in \eq{Rj-dfn} imply that
\be{resolution2-1}
\sum_{j \in \cB} R_j^{(\pi)} = I(U_1^L[\cB];Y,U_1^L[\cB^c]). 
\ee 
Notice that $\bQ$ is a point included in $\cU(W)$. Therefore, using \eq{region-def} and \eq{resolution2-1} we get
$$
\sum_{j \in \cB} Q_j \leq I(U_1^L[\cB];Y_1^L, U_1^L[\cB^c]) = \sum_{j \in \cB} R_j^{(\pi)}.
$$
But since $\cB \subset \AQ$, for any $j \in \cB$, $Q_j > R_j^{(\pi)}$ which is a contradiction. This proves the lemma. 
\end{proof}

For two users $j$ and $j'$, we say that $j'$ is reachable from $j$, if there exists a sequence of users $j_1,j_2,\dots,j_t$, with $j_i \in [\![m]\!]$, such that $j \rightarrow j_0 \rightarrow j_1 \rightarrow \dots \rightarrow j_t \rightarrow j'$. This sequence is also referred to as the path from $j$ to $j'$. If such a path exists, then one can assume, without loss of generality, that $j_1,j_2,\dots,j_t$ are distinct. Clearly, if $j''$ is reachable from $j'$ and $j'$ is reachable from $j$, then $j''$ is reachable from $j$. 

\begin{corollary}
\label{c-path}
For any $j \in \AQ$, at least one element of the complement set $[\![m]\!] - \AQ$ is reachable from $j$. 
\end{corollary}
\begin{proof}
Let $\cB$ be the set of all reachable users from $j$. Then $\cB$ is a closed set in the sense that no element in $\cB^c$ is reachable from any element of $\cB$. It implies that through the permutation $\pi$, all the input bits of $\cB$ appear after all the input bits of $\cB^c$. Then, by \Lref{resolution2}, $\cB$ is not a subset of $\AQ$ which proves the corollary.   
\end{proof}

Now that we have established the necessary notations and derived desired properties of the $m$-tuple rates of the polarization bases, we turn to state the main theorem of this section as follows: 
\begin{theorem}
\label{cover}
Given a point $\bQ \in \cD(W)$ and positive integer $L$, there exists a permutation $\pi$ on the set $[\![mL]\!]$ such that for any index $j \in [\![m]\!]$, we have
$$
|Q_j - R^{(\pi)}_j| \leq \frac{m-1}{L}.
$$
\end{theorem}
\begin{proof}
Let $\pi$ be the permutation that minimizes the Euclidean distance $\left\|\bR^{(\pi)}-\bQ\right\|$ among all the permutations, i.e., for any other permutation $\pi'$:
$$
\left\|\bR^{(\pi)}-\bQ\right\| \leq \left\|\bR^{(\pi')}-\bQ\right\|.
$$
Consider $\AQ$ as defined in \eq{AQ-def}. If $\AQ$ is empty, then it implies that for any $j \in [\![m]\!]$, $Q_j \leq R^{(\pi)}_j$. But both $\bR^{(\pi)}$ and $\bQ$ have the same sum-rate, i.e.,
$$
\sum_{j \in [\![m]\!]} R^{(\pi)}_j = \sum_{j \in [\![m]\!]} Q_j = \cI(W).
$$
Therefore, $\bQ = \bR^{(\pi)}$ and the theorem is proved. In the case that $\AQ$ is non-empty, let $j_1$ be an arbitrary element of $\AQ$. By \Cref{c-path}, there is a path $j_1 \rightarrow j_2 \rightarrow \dots \rightarrow j_t$ with $j_i \in \AQ$, for $i = 2,\dots,t-1$, and $j_t \notin \AQ$. Without loss of generality, we can assume that $j_i$'s are distinct and hence $t \leq m$. For notational convenience, the users are relabeled so that the user $j_i$ is labeled with $i$ for $i \in [\![t]\!]$. 

Define the sequence of real numbers $\left\{a_i\right\}_{i=1}^{t}$, where $a_i = Q_i - R^{(\pi)}_{i}$. We claim that $a_{i-1} - a_i \leq \frac{1}{L}$. The claim will imply that $a_1 - a_t \leq \frac{t-1}{L}$, and since $a_t$ is not a positive number as $t \notin \AQ$, it follows that 
$$
a_1  \leq a_1 - a_t \leq \frac{t-1}{L} \leq \frac{m-1}{L}.
$$

By symmetry, the same set of arguments can be applied to the set $[\![m]\!] - \AQ$, where one can prove that $R^{(\pi)}_{j} - Q_j \leq \frac{m-1}{L}$, for $j \in [\![m]\!] - \AQ$. This will complete the proof of theorem. Hence, it remains to prove the claim about the sequence $\left\{a_i\right\}_{i=1}^{t}$.  

For $i = 2,3,\dots,t$, let the permutation $\pi_i$ be the transposition of $\pi$ with respect to $i-1 \rightarrow i$. 
By \Lref{resolution1}, we have
\begin{align}
\label{resolution3-1}
0 \leq R_{{i-1}}^{(\pi_i)} - R_{i-1}^{(\pi)} &= R_{i}^{(\pi)} - R_{i}^{(\pi_i)} \leq \frac{1}{L}\\
\label{resolution3-2}
R_{j}^{(\pi)} &= R_{j}^{(\pi_i)}\ \text{for}\ {j \neq i-1,i}.
\end{align}
On the other hand, by the choice of $\pi$, we know that
\be{resolution3-3}
\left\|\bR^{(\pi)} - \bQ\right\| \leq \left\|\bR^{(\pi_i)} - \bQ\right\|.
\ee
\eq{resolution3-2} and \eq{resolution3-3} together imply that
\be{resolution3-4}
\begin{split}
(Q_{i-1} -  R_{i-1}^{(\pi)})^2 &+ (Q_i -  R_{i}^{(\pi)})^2 \\
&\leq (Q_{i-1} -  R_{i-1}^{(\pi_i)})^2 + (Q_i -  R_{i}^{(\pi_i)})^2.
\end{split}
\ee
Let $\alpha = R_{{i-1}}^{(\pi_i)} - R_{i-1}^{(\pi)}$, then by \eq{resolution3-1}, \eq{resolution3-4} can be re-written as
$$
a_{i-1}^2 + a_i^2 \leq (a_{i-1} - \alpha)^2 + (a_{i}+\alpha)^2
$$
which can be simplified as
$$
a_{i-1} - a_i \leq \alpha.
$$
By \eq{resolution3-1}, $\alpha \leq \frac{1}{L}$ which completes the proof of claim.
\end{proof}

\begin{corollary}
\label{cover2}
The covering radius associated with the MAC polarization bases of length $L$ is upper bounded by $\frac{(m-1)\sqrt{m}}{L}$.
\end{corollary}
\begin{proof}
Consider an arbitrary point $\bQ$ on the dominant face. Then the permutation $\pi$ exists as in \Tref{cover}, and 
$$
\left\|\bQ - \bR^{(\pi)}\right\|^2 = \sum_{j = 1}^{m} (Q_j - R_j^{(\pi)})^2 \leq \frac{m(m-1)^2}{L^2}.
$$ 
Then, the corollary follows  by the definition of the covering radius in \eq{covering-radius}.
\end{proof}

\noindent
{\bf Remark.} \looseness=-1 
In Theorem 5 of \cite{CISS}, we showed for a two-user MAC that for any point $\bQ$ on the dominant face there exists a pair of achievable rates with a distance at most $\frac{\sqrt{2}}{L}$ from $\bQ$. In other words, the covering radius for the special case of $m=2$ is upper bounded by $\frac{\sqrt{2}}{L}$. This result matches with the result of \Cref{cover2} for $m = 2$.

\section{MAC Polarization and Asymptotic Performance}
\label{sec:six}

In this section, we show that, for all $L=2^l$ and permutations $\pi$, the points $\rp$ defined in \eq{Rj-dfn} are achievable by MAC polar coding. Then the asymptotic performance of the proposed MAC polar codes are discussed.

\subsection{MAC Polarization}

Consider a MAC polarization base of length $L=2^l$ associated with a $\pi$ on $[\![mL]\!]$.  Let $n \geq l$ be a positive integer and $N = 2^n$. We show in \Lref{channel-split-general} that the MAC polar transformations of length $N$ built upon the MAC polarization base of length $L$
can be decomposed into single-user polar transformations, with length $N/L$, of the bit-channels corresponding to the polarization base. In this case, we say that the MAC polar transformation is built upon the polarization base by $n-l$ levels of polarization.

Let $(U_1^N[1],U_1^N[2],\dots,U_1^N[m])$ be the input sequences of a MAC polar transformation associated with $\pi^{(N/L)}$, as defined in Section \ref{Sec:Not}. Let also $Y_1^N$ denote the output sequence.


The following lemma holds for the particular MAC polar transformation defined here and its corresponding bit-channels. This lemma is proved in Appendix A.

\begin{lemma}
\label{channel-split-general}
Given $N = 2^n \geq L$ independent copies of $W$ and the bit-channels defined in \eq{Wi-def-general}, for any $j$ with $1 \leq j \leq N$,
$$
W_{2N}^{(2j-1)} = W_N^{(j)} \boxcoasterisk W_N^{(j)}
$$
and
$$
W_{2N}^{(2j)} = W_N^{(j)} \circledast W_N^{(j)}.
$$
\end{lemma}


\begin{corollary}
\label{recursion-split-general}
For a given $k \in \left[mL\right]$, let $\overline{W}$ denote the single user channel $W_{L}^{(k)}$. Let also $n \geq l$ and $N=2^n$. Then, for any $\frac{(k-1)N}{L}+1 \leq j \leq \frac{kN}{L}$, 
$$
W_{N}^{(j)} = \overline{W}_{\frac{N}{L}}^{\left(j-\frac{(k-1)N}{L}\right)},
$$
where the bit-channels with respect to $\overline{W}$ are the single user bit-channels. 
\end{corollary}
\begin{proof}
The proof is by induction on $n$. The base of induction is trivial for $n=1$. The induction steps are by the fact that channel combining recursion steps in Arikan's polar transformation, as proved in Proposition 3 of \cite{Arikan}, match with \Lref{channel-split-general}.   
\end{proof}

Let $\cT = 2^{-N^{\beta}}\!\!/mN$, where $\beta \,{<}\, \shalf$, and consider the set of good bit-channels $\cG^{(j)}_N(W, \cT)$, for any $j \in [\![m]\!]$, as defined in \eq{good-def}. The next theorem establishes the main result of the channel polarization for the MAC polar transformation.
\begin{theorem}
\label{thm2}
For any $m$-user binary-input discrete MAC $W$ and $j \in [\![m]\!]$, we have
$$
\lim_{N \to \infty} \frac{\bigl|\cG^{(j)}_N(W, \cT)\bigr|}{N} = R_j^{(\pi)}.
$$
\end{theorem}
The proof can be found in Appendix A.

\subsection{Asymptotic Performance}

The code construction and encoding process are done with respect to the set of good bit-channels $\cG^{(j)}_N(W, \cT)$, as explained in Section\,\ref{sec:construction}. The multi-user successive cancellation decoder is also done, as explained in Section\,\ref{sec:decoding}, using a low complex implementation that captures the essence of bit-channel recursions established in \Lref{channel-split-general} and \Cref{recursion-split-general}. 


For the asymptotic performance of the proposed MAC polar coding scheme, the following theorem follows, similar to \cite{Arikan} and \cite{AT}.
\begin{theorem}
\label{thm-main}
For any $\beta \,{<}\, \shalf$, any $m$-user MAC $W$, any $\epsilon > 0$ and any point $\bQ$ on $\cD(W)$, there exists a family of polar codes that approaches a point on the dominant face within distance $\epsilon$ from $\bQ$. Furthermore, the average probability of frame error under successive cancellation decoding is less than $2^{-N^{\beta}}$, where $N$ is the code block length for each of the users.
\end{theorem}
\begin{proof}
The choice of $L$ for the polarization base depends on $\epsilon$. Fix $L = 2^l$ such that $\frac{(m-1)\sqrt{m}}{L} < \epsilon$. Then, by \Cref{cover2} there exists $\pi$ such that the distance between $\bQ$ and $\rp$ is less than $\epsilon$. Fixing $L$ and $\pi$, for any block length $N = 2^n \geq L$, we construct the polar code for the $j$-th user with respect to the set of good bit-channels $\cG_N^{(j)}(W,\cT)$, where $\cT = 2^{-N^{\beta}}\!\!/mN$, defined in \eq{good-def}. The rest of bit-channels are set to fixed values selected according to independent and uniform distributions and then revealed to both the encoder and the decoder. Then, by \Tref{thm2}, the $m$-tuple of rates approach $\rp$ as $N$ goes to infinity. The probability of frame error is calculated as the average over all the possible information vectors as well as the frozen values and is bounded by the sum of the Bhattacharrya parameters of the selected good bit-channels. This follows similar to ~\cite[Proposition 2.]{Arikan}. Hence, by definition of $\cG_N^{(j)}(W,\cT)$, the average frame error probability is less than $2^{-N^{\beta}}$ for the multi-user successive cancellation decoding, as discussed in Section\,\ref{sec:decoding}. 
\end{proof}

\section{Related Works and Comparisons  }
\label{sec:seven}

\subsection{Comparison with Other MAC Polar Coding Schemes}

In a related work, the authors of \cite{STY,AT2} propose a framework for MAC polarization by extending the notion of channel splitting from the single user case to the two user case \cite{STY}, and then to the $m$-user case \cite{AT2}. Alternatively, our MAC polar coding scheme in this paper views the MAC polarization as a single user channel polarization by considering all the possible decoding orders. First, we compare the two schemes in terms of decoding complexity and then capacity achieving property.

Both our scheme and that of \cite{STY,AT2} use a successive cancellation decoding, originally proposed by Ar{\i}kan in \cite{Arikan}. As discussed in Section\,\ref{sec:decoding}, we have to combine the low-complex Ar{\i}kan's decoder with a basic decoder for the polarization base. As a result, we have an extra additive term $O(mN 2^{mL})$ in the complexity which is dominated by $O(mN (n-l+1))$, as $N$ goes to infinity, while $L$ and $m$ are assumed to be constant. On the other hand, in the scheme proposed in \cite{STY} and extended in \cite{AT2}, the likelihood of a vector of length $m$ needs to be tracked along the decoding trellis. Therefore, instead of a simple likelihood ratio, a vector of length $2^m$ for the probability of all $2^m$ possible cases has to be computed recursively. This increases the decoding complexity by a factor of $2^m$. In fact, the decoding complexity is still $O(N\log N)$, but if we look at the actual number of operations needed to complete the decoding, the decoding complexity of the scheme in \cite{STY,AT2} is $\frac{2^m}{m}$ times more than the decoding complexity of our scheme, asymptotically.  

The two schemes can be also compared in terms of capacity-achieving property. In \Tref{thm-main}, we proved that all the points on the dominant face of the uniform rate region can be achieved with arbitrary fine resolution. As pointed out before, the scheme proposed in \cite{STY,AT2} does not necessarily achieve the whole uniform rate region. It is only guaranteed that one point on the dominant face is achievable. There actually exist examples of two user MACs such that the scheme proposed in \cite{STY} achieves only one point on the dominant face. 

\subsection{Polar Coding for Interference Networks}

A polar coding scheme for interference networks is introduced in \cite{WS}. The authors of \cite{WS} develop a polar coding method that achieves the Han-Kobayashi inner bound for the two-user interference networks. An intermediate step in doing so is MAC polar coding for the two-user case where the dual of the Ar{\i}kan's scheme of \cite{Arikan3} is used in the context of multiple access channels. In order to extend the results to the general case of $m$-user $k$-receiver interference network, the intermediate step would be to generalize the scheme of \cite{Arikan3}, based on the monotone chain rule of mutual information, to $m$-user multiple access channels. This problem is called polar splitting in \cite{WS} and it is conjectured that an induction on the number of users $m$ would lead to such result. We briefly recap the induction procedure of \cite[Proposition 2.]{WS}. We further point out that in order to complete the induction one has to increase the input alphabet of the resulting MACs exponentially in each step of the induction. Although it is pointed out in the proof of~\cite[Proposition 2.]{WS} that the resulting sub-problems in each step of the induction are considered as MACs, it might be possible to relax this condition in order to overcome the complexity issue. To this end, the statement of the~\cite[Proposition 2.]{WS} needs to be also relaxed in order to cover more general scenarios than MACs and then the induction will follow without the need to exponentially increase the output alphabet sizes at each step of the induction. 


Consider $L$ uses of $W$, with inputs $U_1^L[j] \in \sX^L$, for $j \in [\![m]\!]$, a polar transformation of length $L$, and output $Y_1^L \in \sY^L$. Let also $(Q_1,Q_2,\dots,Q_m)$ denote the $m$-tuple of rates on the dominant face to be approximated. Let $i$ increases from $1$ to $L$ and consider quantities 
\be{WS-1}
\frac{1}{L} I(U_1^L[\cJ];Y_1^L,U_1^i[1]),
\ee
for each $\cJ \subseteq [\![m]\!]\setminus\left\{1\right\}$. As $i$ increases, each mutual information in \eq{WS-1} increases by at most $\frac{1}{L}$ in each step. Also, it is shown that there exists an $i$ such that for at least one $\cJ \subseteq  [\![m]\!]\setminus\left\{1\right\}$ the following is violated
\be{WS-2}
\frac{1}{L} I(U_1^L[\cJ];Y_1^L,U_1^i[1]) < Q(\cJ),
\ee
where $Q(\cJ) = \sum_{j \in \cJ} Q_j$. The first $i_0$ and $\cJ_0$ for which this happens are considered. Then the problem of approaching the rates on the dominant face, also referred to as the MAC rate approximation problem for $W$, is split into two separate problems with MACs $W_1$ and $W_2$ having smaller number of users but exponentially larger input and output alphabets, as elaborated below. 

Let $m_0 = |\cJ_0|$. The first MAC $W_1 : (\sX^L)^{m_0} \rightarrow \sY^L \times {\sX}^{i_0}$ is deduced from $L$ uses of $W$ having the input sequence $U_1^L[\cJ_0]$ and the output sequence $Y_1^L$ together with $U_1^{i_0}[1]$, while $U_{i_0+1}^{L}[1]$ and $U_1^L[\cJ_0^c \setminus\left\{1\right\}]$ are regarded as noise. The second MAC $W_2 : \sX^{L-i_0}\times ({\sX^L})^{L-m_0-1} \rightarrow \sY^L \times \sX^{i_0}\times ({\sX^L})^{m_0}$ is also deduced from $L$ uses of $W$ having the input sequence $U_{i_0+1}^L[1]$ and $U_1^L[\cJ_0^c \setminus\left\{1\right\}]$ and the output sequence $Y_1^L$, $U_1^{i_0}[1]$ and $U_1^L[\cJ_0]$. 


The goal of the induction procedure of \cite[Proposition 2.]{WS} is to satisfy the rate approximation requirements for the $m$ users one by one, i.e., in the first induction step the design rate $Q_1$ of the first user is approximated and the rest of the users are split into two groups of $\cJ_0$ and $\cJ_0^c \setminus\left\{1\right\}$ corresponding to MACs $W_1$ and $W_2$. The induction procedure of \cite[Proposition 2.]{WS} results in a super exponentially large $L$ at the end, in terms of the number of users, and a particular structure on the order of input bits. In contrast, we arrive at the solution for the rate approximation problem for all the users at once using a fixed $L$. To this end, we consider all the possible permutations of the input bits and find a sequence of permutation transpositions that results in the desired solution, as discussed in Section \ref{sec:five}, and provided a concrete proof. We further derived an upper bound on the covering radius in terms of the number of users $m$ and the length of the polarization base $L$. In a sense, this provides a precision measure for the rate approximation problem of~\cite{WS} which is not addressed in~\cite{WS}. 

\subsection{Rate Splitting Method}

An alternative method to achieve the uniform rate region of multiple access channels is by rate splitting between the users \cite{RU, GRU}. It is shown in \cite{GRU} that the encoding/decoding problem for any asynchronous $m$-user discrete memoryless MAC can be reduced to corresponding problems for at most $2m-1$ single-user channels. This is used in \cite[Appendix A]{STY}, where it is shown how polar coding can be adapted in this context. In particular, for achieving the uniform rate region, any target $m$-tuple of rates on the dominant face can be turned into a corner point of a $2m-1$-user MAC, where the input distributions of the $2m-1$ users are no longer uniform. However, binary polar codes can only be used for achieving the symmetric capacity of binary-input channels, where the input distribution is uniform. 

The solution suggested in \cite[Appendix A]{STY} to deal with the non-uniform distributions is to use Gallager's method \cite[p. 208]{gallager}, where a mapper from a larger alphabet $\sX'$ with uniform distribution to a smaller alphabet $\sX$ is used to approximate a non-uniform distribution on $\sX$. Then polar coding over a discrete channel with the input alphabet $\sX'$ is used to arrive at the solution. However, in order to achieve a finer resolution on the dominant face, the resulting alphabet $\sX'$ would become larger. In particular, to the best of our knowledge, a solution for MAC polar coding with rate splitting method, where binary polar coding, or at least polar coding over a channel with fixed alphabet, is used at the users does not still exist.

\section{Concluding Remarks}
\label{sec:eight}

In this paper, we considered the problem of designing polar codes for transmission over general $m$-user multiple access channels. The key observation behind our work is to view the polar transformations for the $m$-users across the multiple access channel as a unified MAC polar transformation. We showed that the MAC polar transformation of length $N$ with a certain decoding order can be split into $mL$ single user polar transformations, built upon a MAC polarization base of length $L$.  We also proved that the covering radius of the dominant face is upper bounded by $\frac{(m-1)\sqrt{m}}{L}$. Therefore, by letting $L$ to grow large, we are able to achieve the entire uniform rate region with the constructed MAC polar code. Moreover, our construction allows low complexity multi-user successive cancellation decoding with an asymptotic decoding complexity of $O(N \log N)$, similar to that of single-user polar decoding. 

From a theoretical point of view, the solution for achieving the uniform rate region with MAC polar coding presented in this paper has the advantage of being useful for other multi-user communication set-ups. In particular, it can be used as an alternative building block for polar coding scheme of \cite{WS} to achieve the Han-Kobayashi inner bound for the interference networks. 


Since each user's code can be regarded as single user polar code, existing methods for improving the performance of single user polar codes can be applied on top of our proposed MAC polar codes. For instance, the individual polar codes can be concatenated with other block codes \cite{RS-polar} or can be made systematic \cite{Arikan4}.  Furthermore, the MAC polar code construction can also be modified by relaxing the polarization of certain bit-channels as proposed in \cite{ElKh1503:Relaxed}  to reduce both the encoding and decoding  latencies and computational complexities. List decoding of polar codes \cite{TV} can aslo be used to boost the finite-length performance.

\bibliographystyle{IEEEtran}
\bibliography{polar2}

\begin{thebibliography}{10}
\providecommand{\url}[1]{#1}
\csname url@samestyle\endcsname
\providecommand{\newblock}{\relax}
\providecommand{\bibinfo}[2]{#2}
\providecommand{\BIBentrySTDinterwordspacing}{\spaceskip=0pt\relax}
\providecommand{\BIBentryALTinterwordstretchfactor}{4}
\providecommand{\BIBentryALTinterwordspacing}{\spaceskip=\fontdimen2\font plus
\BIBentryALTinterwordstretchfactor\fontdimen3\font minus
  \fontdimen4\font\relax}
\providecommand{\BIBforeignlanguage}[2]{{%
\expandafter\ifx\csname l@#1\endcsname\relax
\typeout{** WARNING: IEEEtran.bst: No hyphenation pattern has been}%
\typeout{** loaded for the language `#1'. Using the pattern for}%
\typeout{** the default language instead.}%
\else
\language=\csname l@#1\endcsname
\fi
#2}}
\providecommand{\BIBdecl}{\relax}
\BIBdecl

\bibitem{Arikan}
E.~Arikan, ``Channel polarization: A method for constructing capacity-achieving
  codes for symmetric binary-input memoryless channels,'' \emph{IEEE
  Transactions on Information Theory}, vol.~55, no.~7, pp. 3051--3073, 2009.

\bibitem{MV}
H.~Mahdavifar and A.~Vardy, ``Achieving the secrecy capacity of wiretap
  channels using polar codes,'' \emph{IEEE Transactions on Information Theory},
  vol.~57, no.~10, pp. 6428--6443, 2011.

\bibitem{Arikan2}
E.~Arikan, ``Source polarization,'' \emph{Proceedings of IEEE International
  Symposium on Information Theory}, pp. 899--903, 2010.

\bibitem{Ab}
E.~Abbe, ``Randomness and dependencies extraction via polarization,''
  \emph{Proceedings of Information Theory and Applications Workshop (ITA)}, pp.
  1--7, 2011.

\bibitem{MH}
M.~Mondelli, S.~H. Hassani, I.~Sason, and R.~Urbanke, ``Achieving marton's
  region for broadcast channels using polar codes,'' \emph{Proceedings of IEEE
  International Symposium on Information Theory}, pp. 306--310, 2014.

\bibitem{BICMpolar}
H.~Mahdavifar, M.~El-Khamy, J.~Lee, and I.~Kang, ``Polar coding for
  bit-interleaved coded modulation,'' \emph{Vehicular Technology, IEEE
  Transactions on}, 2015.

\bibitem{A}
R.~Ahlswede, ``Multi-way communication channels,'' \emph{Proceedings of 2nd
  International Symposium on Information Theory}, Tsahkadsor, Armenia SSR,
  1971.

\bibitem{L}
H.~H.-J. Liao, ``Multiple access channels,'' Ph.D. dissertation, Univ. Hawaii,
  Honolulu, 1972.

\bibitem{RU}
B.~Rimoldi and R.~Urbanke, ``A rate-splitting approach to the gaussian
  multiple-access channel,'' \emph{IEEE Transactions on Information Theory},
  vol.~42, no.~2, pp. 364--375, 1996.

\bibitem{GRU}
A.~J. Grant, B.~Rimoldi, R.~L. Urbanke, and P.~A. Whiting, ``Rate-splitting
  multiple access for discrete memoryless channels,'' \emph{IEEE Transactions
  on Information Theory}, vol.~47, no.~3, pp. 873--890, 2001.

\bibitem{STY}
E.~\c{S}a\c{s}o\u{g}lu, E.~Telatar, and E.~Yeh, ``Polar codes for the two-user
  binary-input multiple-access channel,'' \emph{IEEE Transactions on
  Information Theory}, vol.~59, no.~10, pp. 6583--6592, 2013.

\bibitem{AT2}
E.~Abbe and E.~Telatar, ``Polar codes for the m-user multiple access channel,''
  \emph{IEEE Transactions on Information Theory}, vol.~58, no.~8, pp.
  5437--5448, 2012.

\bibitem{S}
S.~Onay, ``Successive cancellation decoding of polar codes for the two-user
  binary-input {MAC},'' \emph{Proceedings of IEEE International Symposium on
  Information Theory}, pp. 1122--1126, July 7-12, 2013.

\bibitem{CISS}
H.~Mahdavifar, M.~El-Khamy, J.~Lee, and I.~Kang, ``Techniques for polar coding
  over multiple access channels,'' \emph{Proceedings of 48th Annual Conference
  on Information Sciences and Systems (CISS)}, Princeton, NJ, USA, March
  19-21,2014.

\bibitem{Arikan3}
E.~Arikan, ``Polar coding for the {Slepian-Wolf} problem based on monotone
  chain rules,'' \emph{Proceedings of IEEE International Symposium on
  Information Theory}, pp. 566--570, 2012.

\bibitem{WS}
L.~Wang and E.~Sasoglu, ``Polar coding for interference networks,'' \emph{arXiv
  preprint arXiv:1401.7293}, 2014.

\bibitem{AT}
E.~Arikan and E.~Telatar, ``On the rate of channel polarization,''
  \emph{Proceedings of IEEE International Symposium on Information Theory}, pp.
  1493--1495, 2009.

\bibitem{Korada}
S.~B. Korada, ``Polar codes for channel and source coding,'' Ph.D.
  dissertation, {\'E}cole Polytechnique F{\'e}d{\'e}ral De Lausanne, 2009.

\bibitem{KSU}
S.~B. Korada, E.~\c{S}a\c{s}o\u{g}lu, and R.~Urbanke, ``Polar codes:
  Characterization of exponent, bounds, and constructions,'' \emph{IEEE
  Transactions on Information Theory}, vol.~56, no.~12, pp. 6253--6264, 2010.

\bibitem{modern}
T.~Richardson and R.~Urbanke, \emph{Modern coding theory}.\hskip 1em plus 0.5em
  minus 0.4em\relax Cambridge University Press, 2008.

\bibitem{compound}
H.~Mahdavifar, M.~El-Khamy, J.~Lee, and I.~Kang, ``Compound polar codes,''
  \emph{Proceedings of Information Theory and Applications Workshop (ITA)}, pp.
  1--6, San Diego, California, 2013.

\bibitem{gallager}
R.~G. Gallager, \emph{Information theory and reliable communication}.\hskip 1em
  plus 0.5em minus 0.4em\relax Springer, 1968, vol.~2.

\bibitem{RS-polar}
H.~Mahdavifar, M.~El-Khamy, J.~Lee, and I.~Kang, ``Performance limits and
  practical decoding of interleaved {Reed-Solomon} polar concatenated codes,''
  \emph{IEEE Transactions on Communications}, vol.~62, no.~5, pp. 1406--1417,
  2014.

\bibitem{Arikan4}
E.~Ar{\i}kan, ``Systematic polar coding,'' \emph{IEEE communications letters},
  vol.~15, no.~8, pp. 860--862, 2011.

\bibitem{ElKh1503:Relaxed}
M.~El-Khamy, H.~Mahdavifar, G.~Feygin, J.~Lee, and I.~Kang, ``Relaxed channel
  polarization for reduced complexity polar coding,'' in \emph{2015 IEEE
  Wireless Communications and Networking Conference (WCNC 2015)}, New Orleans,
  USA, Mar. 2015, pp. 219--224.

\bibitem{TV}
I.~Tal and A.~Vardy, ``List decoding of polar codes,'' \emph{Proceedings of
  IEEE International Symposium on Information Theory}, pp. 1--5, 2011.

\end{thebibliography}

\begin{appendix}

\textit{Proof of \Lref{resolution1}:}
By definition, $\pi'$ is the result of $\pi$ by swapping $\pi(i_1)$ and $\pi(i_2)$, where $\pi(i_2) = \pi(i_1) +1$. Therefore, by definition of $I_i^{(\pi)}$ in \eq{Iji-dfn}, $I_i^{(\pi)} = I_i^{\pi'}$ for $i \neq i_1,i_2$. This implies that the values of $R_j^{\pi}$ do not change during the specified transposition, for $j \neq j_1,j_2$. 

For the second part, 
\begin{equation}
\begin{split}
\label{resolution-eq1}
R_{j_1}^{(\pi')} - R_{j_1}^{(\pi)} &= \frac{\sum_{i=(j_1-1)L+1}^{j_1 L} I_i^{(\pi')}}{L} - \frac{\sum_{i=(j_1-1)L+1}^{jL} I_i^{(\pi)}}{L} \\
& = \frac{I_{i_1}^{(\pi')} - I_{i_1}^{(\pi)}}{L}.
\end{split}
\end{equation}
Also,
\be{resolution-eq2}
0 \leq I_{i_1}^{(\pi)} \leq I_{i_1}^{(\pi')} \leq 1.
\ee
\eq{resolution-eq1} and \eq{resolution-eq2} together imply that
$$
0 \leq R_{j_1}^{(\pi')} - R_{j_1}^{(\pi)} \leq \frac{1}{L}.
$$
The other inequality also follows similarly or simply by observing that $\sum R_j^{(\pi)} = \sum R_j^{(\pi')} = \cI(W)$, by \Lref{sum-rate}.
\hspace*{\fill}~\QED\par\endtrivlist\unskip

\textit{Proof of \Lref{channel-split-general}:}
Let $d^{2N}_{1,o}$ and $d^{2N}_{1,e}$ denote the even-indexed and odd-indexed sub-sequences, respectively. Then
\begin{align*}
W_{2N}^{(2i)}  = & \frac{1}{2^{2mN-1}} \sum_{d_{2i+1}^{2mN}} \widetilde{W}_N\bigl(y^N_{1,o} | d^{2N}_{1,o} \oplus d^{2N}_{1,e} \bigr) 
\widetilde{W}_N\bigl(y^N_{1,e} | d^{2N}_{1,e} \bigr) \\
= & \frac{1}{2} \frac{1}{2^{mN-1}} \sum_{d_{2i+1,e}^{2N}} \widetilde{W}_N\bigl(y^N_{1,e} | d^{2N}_{1,e} \bigr). \\
& \frac{1}{2^{mN-1}} \sum_{d_{2i+1,o}^N} \widetilde{W}_N\bigl(y^N_{1,o} | d^{2N}_{1,o} \oplus d^{2N}_{1,e} \bigr) \\
= & W_N^{(i)} \circledast W_N^{(i)}
\end{align*} 
where we used the definitions of bit-channels in \eq{Wi-def-general} and the channel combining operation in \eq{channel-com1} and \eq{channel-com2} together with the recursive structure of polar transformation. The other equation can be derived using the similar arguments.
\hspace*{\fill}~\QED\par\endtrivlist\unskip

\textit{Proof of \Tref{thm2}:}
By \Cref{recursion-split-general}, 
$$
\cG^{(j)}_N(W, \cT) = \bigcup_{i \in S^{(j)}_L} \cG_{\frac{N}{L}}(W^{(i)}_{L}, \cT),
$$ 
where $S^{(j)}_L$ is defined in \eq{S-def}.
Therefore, channel polarization theorem for the single user case \cite{AT} imply that
\begin{align*}
\lim_{N \to \infty} \frac{\bigl|\cG^{(j)}_N(W, \cT)\bigr|}{N} &\geq \frac{1}{L} \sum_{i \in S^{(j)}_L} \lim_{N \to \infty} \frac{\bigl|\cG_{\frac{N}{L}}(W^{(i)}_{L},\cT)\bigr|}{N/L} \\
& = \frac{1}{L} \sum_{i \in S^{(j)}_L} I_i^{(\pi)} = R_j^{(\pi)}.
\end{align*}
Since the fraction of all good bit-channels can not exceed the sum-rate $\cI(W) = \sum_{j \in [m]} R_j^{(\pi)}$, the above inequality should be equality. This completes the proof of theorem.
\hspace*{\fill}~\QED\par\endtrivlist\unskip

\end{appendix}

\end{document}